\documentclass[conference]{IEEEtran}
\usepackage{blindtext, graphicx}
\usepackage{amsmath, amsthm, amssymb, algorithmic, array, mathtools}
\usepackage{breqn}
\usepackage{empheq}
\usepackage{verbatim}
\usepackage[noadjust]{cite}
\usepackage{relsize}

\allowdisplaybreaks  

\newtheorem{proposition}{Proposition}

\DeclareMathOperator*{\argmax}{\arg\!\max}

\begin{document}
\title{\huge Learning with Finite Memory for \\ Machine Type Communication\vspace{-.2cm}}

\author{\IEEEauthorblockN{Taehyeun Park$^1$ and Walid Saad$^1$}
\IEEEauthorblockA{\small $^1$Wireless@VT, Bradley Department of Electrical and Computer Engineering, Virginia Tech, Blacksburg, VA, USA,\\ Emails:\{taehyeun, walids\}@vt.edu \vspace{-.5cm}}}

\maketitle

\begin{abstract}
Machine-type devices (MTDs) will lie at the heart of the Internet of Things (IoT) system. A key challenge in such a system is sharing network resources between small MTDs, which have limited memory and computational capabilities. In this paper, a novel learning \emph{with finite memory} framework is proposed to enable MTDs to effectively learn about each others message state, so as to properly adapt their transmission parameters. In particular, an IoT system in which MTDs can transmit both delay tolerant, periodic messages and critical alarm messages is studied. For this model, the characterization of the exponentially growing delay for critical alarm messages and the convergence of the proposed learning framework in an IoT are analyzed. Simulation results show that the delay of critical alarm messages is significantly reduced up to $94\%$ with very minimal memory requirements. The results also show that the proposed learning with finite memory framework is very effective in mitigating the limiting factors of learning that prevent proper learning procedures.
\end{abstract}

\section{Introduction}
The Internet of Things (IoT) is seen as the most promising networking technology of the coming decade \cite{evans}. In the IoT, a massive number of machine type devices (MTDs), such as sensors, wearables, or even mundane objects \cite{saad}, will be endowed with large-scale wireless communication capabilities. Such MTDs will be critical for applications such as smart metering, smart agriculture, and healthcare. However, to enable such massive IoT communication, many challenges must be overcome ranging from network modeling to resource management \cite{ratasuk, popov, latency, dhil1, dhil2}. First and foremost, wireless communication for the IoT must take into account stringent resource constraints, massive scale, and high heterogeneity \cite{ratasuk}. Additionally, it must satisfy strict system requirements, such as enhanced security, high reliability, and efficient resource management \cite{ratasuk}.\\
\indent The MTDs will have different types of messages of varying characterizations, such as time sensitivity, urgency, and delay tolerance \cite{popov}. Some messages will be delay tolerant and not urgent and, thus, they will not require retransmissions or short delays. However, there also are urgent, delay intolerant messages, which require retransmissions and ultra low latency \cite{latency}. Thus, it is critical to develop reliable wireless communication protocols for the IoT that can account for such message heterogeneity as well as for the stringent computational and memory constraints of the MTDs.\\
\indent There are many works on designing, evaluating, and analyzing various wireless communication solutions for the MTDs in an IoT \cite{ratasuk, latency, dhil1, dhil2, popov}. Some of these works, such as in \cite{latency, dhil1, dhil2}, focus on improving latency, reliability, system throughput, and transmit power, which are of critical importance in an IoT. Other works, such as in \cite{popov}, characterize various types of messages present in an IoT. A framework for machine-to-machine (M2M) communication using various multiple access schemes is designed and analyzed to optimize the transmit power \cite{dhil2} and to maximize the system throughput \cite{dhil1}. A wireless communications system for an industrial control application that satisfies both low latency and high reliability requirements is developed in \cite{latency}. In \cite{popov}, the messages used in the smart metering application are classified based on maximum latency, payload size, and frequency, and the integration of smart metering into the current wireless cellular systems is investigated. However, none of this existing body of works considers the heterogeneity in message properties and requirements while improving the performance metrics, such as latency and system throughput. Moreover, works such as in \cite{ratasuk}, \cite{latency}, and \cite{dhil1} assume MTDs to have infinite computational and memory requirements and thus are able to operate similarly to classical human type devices such as smartphones. Naturally, such an assumption does not hold for M2M.\\
\indent To cope with the resource limitations of MTDs with respect to computational power and memory, \emph{sequential learning} models can be used to enable the MTDs to learn to optimize their transmission parameters in a distributed manner. In particular, the authors in \cite{cover} and \cite{fm} introduce a finite memory sequential learning scheme using which a number of agents try to learn an underlying binary state of the world in a sequence, such as in a line, with only finite memory. This is in contrast to conventional sequential learning such as in \cite{infocas} and \cite{herd}, which requires the agents to have a complete knowledge of other agents. In both cases, the agents will have to observe the independent, informative \emph{private signals} about the underlying binary state \cite{cover, fm, infocas, herd}. Although conventional sequential learning, which requires an infinite memory, performs better, the finite memory sequential learning is more appropriate for MTD communication. However, existing related works \cite{cover, fm, infocas, herd} are tailored towards economic or social applications and thus, they cannot be readily applied to an IoT setting. Indeed, such works \cite{cover, fm, infocas, herd} do not consider the resource constraints of MTDs, the system properties of an IoT, and the realistic limitations of communication system.\\
\indent The main contribution of this paper is to introduce a novel learning framework using which the MTDs can learn their transmission parameters, under \emph{stringent memory and computational constraints}. In particular, we consider an IoT system in which memory-constrained, low capability MTDs can transmit both delay-tolerant, periodic messages, and real-time, critical alarm message. For this system, the proposed framework enables the MTDs to autonomously learn to adjust their transmission parameters, such as selected transmission code, for high reliability of delivering critical alarm messages. In particular, using the proposed approach, each MTD can properly learn about the presence of alarm messages with only a partial, finite knowledge of other agents. For the studied framework, we show that it can drastically mitigate the effects of limiting factors of learning, which prevent proper sequential learning in an IoT. Simulation results show that the delay of critical alarm message can be significantly reduced, and the number of MTDs that learn the true, underlying binary state is greatly increased. In particular, the results show that the proposed learning with finite memory framework reduces the delay of urgent alarm messages up to $94$\% with minimal memory, while reducing the system throughput by only $17$\%. To our best knowledge, this paper is the \emph{first to develop a learning with finite memory framework within the context of the Internet of Things and machine type communications.}\\
\indent The rest of this paper is organized as follows. Section II presents the system model. In Section III, we introduce the learning with finite memory framework. Section IV analyzes the simulation results while Section V draws conclusions.\\ \vspace{-.5cm}
\section{System Model}
Consider the uplink of a wireless IoT system consisting of one base station (BS) located at the center of a geographical area in which $N$ MTDs are deployed. For the communication between an MTD and the BS, we consider a time slotted system with time slot duration of $\tau$ using a code division multiple access (CDMA) scheme with code length $l$. CDMA is chosen here since it has been recently shown to be promising for supporting a high IoT message arrive rate in \cite{dhil2}. We focus on uplink transmissions during which the MTDs transmit their data to the BS using one out of a fixed number $C = (2^l - 1)$ of binary spreading codes. As in \cite{dhil2}, we assume that there is no coordination in choosing a random code and thus, the chosen codes are not necessarily orthogonal or unique. If more than one MTD use a given code, it would be impossible for the BS to distinguish the messages and thus, all transmissions using the given code will fail. A transmission is considered to be successful if it is the only MTD using a given code. Furthermore, the messages transmitted to the BS are assumed to be short such that, when successful, the transmission can be completed in one period $\tau$. The MTDs are said to be \emph{active} if they have a message to transmit, otherwise, they will be considered \emph{inactive} and will not transmit in that given slot.\\
\indent We let $S$ be the random variable capturing the number of MTDs that have a successful transmission in a given slot. Then, for some number of active MTDs $n$, $0 \leq n \leq N$, the expected value of $S$ will be:
\begin{equation}
\mathbb{E}\left[S\mid n, C\right] = \sum\limits_{s = 1}^{S_{\textrm{max}}} s \Pr(S = s) \binom{n}{s}, \label{systhru}
\end{equation}
where $S_{\textrm{max}}$ is the maximum number of MTDs with successful transmission given $n$ and $C$. $S_{\textrm{max}}$ is defined as

\begin{equation}
    S_{\textrm{max}}=\left\{
                \begin{array}{ll}
                  n \hfill &\text{if} \  n \leq C,\\
                  (C-1) \hfill &\text{if} \  n > C.
                \end{array}
              \right. \label{smax}
\end{equation}
Since it is difficult to directly compute $\Pr(S = s)$, we will first compute $\Pr(S \geq s)$ and use the fact that 
\begin{equation}
\Pr(I = i) = \Pr(I \geq i) - \sum_{h = 1}^{\infty} \Pr(I = (i + h)) \label{fact1}
\end{equation}
for some discrete random variable $I$ with sample space $\Omega$ and $i \in \Omega$. The probability of having at least $s$ MTDs with successful transmission $\Pr(S \geq s)$ will be:
\begin{subequations}
\begin{empheq}[left={\Pr(S \geq s)=}\empheqlbrace]{alignat=2}
& \textstyle \prod\limits_{j=0}^{s-1} \frac{C-j}{C}\left(\frac{C-s}{C}\right)^{\left(n-s\right)} &\textrm{if} \ s < C, s \leq n, \label{eq1}\\
& \textstyle \prod\limits_{j=0}^{s-1} \frac{C-j}{C} &\textrm{if} \ s = C = n, \label{eq2}\\
& 0  &\textrm{otherwise.} \nonumber
\end{empheq}
\end{subequations}
\eqref{eq1} is the probability that $s$ active MTDs transmit successfully using $s$ codes without repetition while the remaining $(n-s)$ MTDs transmit using the remaining $(C-s)$ codes, some of which may transmit successfully. \eqref{eq2} accounts for the special case in which $s = C = n$. Using \eqref{fact1}, \eqref{eq1}, and \eqref{eq2}, the probability of having $s$ MTDs with successful transmission $\Pr(S = s)$ will be:
\begin{equation}
\Pr(S = s) = \Pr(S \geq s)  - \sum\limits_{i = 1}^{S_{\textrm{max}} - s}\binom{n-s}{i}\Pr(S = s + i),\label{eq3}
\end{equation}
where $\binom{n-s}{i}$ accounts for different combination of choosing $i$ active MTDs, which transmit successfully, out of the remaining $(n-s)$ active MTDs with uncertain transmissions.\\
\indent In our IoT model, we consider different types of messages that the MTDs can transmit. Although the message types are application specific, we consider two general types of IoT messages.
\begin{itemize}
\item \emph{Periodic messages} are update messages from the MTDs to the BS. Examples include meter readings, environment observations, and system status reports. This type of messages is delay tolerant and not of critical importance. Such messages contain time sensitive data that may not be useful if outdated. Therefore, retransmitting the same periodic message will not be allowed after failure. The MTDs will transmit periodic messages every $T\tau$.
\item \emph{Alarm messages} are critical messages triggered by a rarely occurring abnormality. Examples of alarm messages include system failure reports, fire or gas leakage, and power outages. This type of messages is highly delay intolerant and of critical importance. It is crucial to quickly notify the BS of this abnormality and thus, retransmissions are necessary. If there is an abnormality, an MTD will transmit the alarm message to the BS every $\tau$ to constantly update the status of the system.
\end{itemize}

\indent Since abnormalities are rare occurrences, we assume that there will not be more than one alarm message at any given time. The time slot for the first periodic message is uniformly distributed over the set $\{1, 2, \cdots, T\}$ such that approximately an equal number of MTDs will transmit the periodic messages in any given time slot. Therefore, the expected number of active MTDs in a given time slot is $\frac{N}{T}$. If an MTD has an alarm message, the probability of successful transmission $p_s$ for the MTD given that there are $(N-1)$ other MTDs possibly with periodic messages will be:\vspace{-.1cm}
\begin{equation}
p_s = \left(\frac{C-1}{C}\right)^{\frac{(N-1)}{T}}. \label{psucc}
\end{equation}
\indent In our model, delay is defined as the time taken until the first successful transmission and thus, the random variable $D$ for delay can be modeled as a geometric random variable with parameter $p_s$. Therefore, the expected value of delay $D$ until the BS is first notified of an abnormality via an alarm message will be:\vspace{-.1cm}
\begin{equation}
\mathbb{E}[D] = \left(\frac{C}{C-1}\right)^{\frac{(N-1)}{T}}. \label{delay}
\end{equation}
\indent As the number of MTDs, $N$, increases, the probability of successful transmission $p_s$ decreases and the expected delay $\mathbb{E}[D]$ increases exponentially. Therefore, in a massive-scale IoT in which $N$ is very large, the expected delay for alarm messages will be very high, which is inacceptable for the urgent, alarm messages, as it will lead to detrimental effects on the overall IoT system. Hence, when there is an alarm message, the MTDs transmitting the periodic messages must adjust their transmission parameters to satisfy the low delay and the high probability of successful transmission requirements of an alarm message.\\
\indent One simple solution is to reserve one of the $C$ codes only for the alarm messages all the time. However, since the alarm messages are a rare occurrence, this can be a waste of code resources in a resource-limited IoT. Moreover, as seen in \eqref{systhru} and \eqref{smax}, reducing the number of codes from $C$ to $(C-1)$ will result in a significant reduction in system throughput. Therefore, a novel framework to minimize the delay of an alarm message without significantly reducing the system throughput is necessary to improve the overall performance of an IoT.
\section{Learning Algorithm}
The MTDs must adjust their communication parameters, in terms of code selection, whenever there is an alarm message. For instance, by never using a reserved code, the system throughput in \eqref{systhru} will be reduced. Therefore, an ideal solution is to enable the MTDs to adjust their communication parameters only once they become aware of the alarm message. Therefore, there is a need for a learning framework using which MTDs can learn of the presence of an alarm message and, subsequently, adapt their transmissions. Such learning will allow the delay of an alarm message \eqref{delay} to be greatly reduced without jeopardizing network performance, such as system throughput in \eqref{systhru}.\\
\indent Any learning algorithm used by the MTDs to learn the presence of an alarm message must satisfy certain properties to account for the unique nature of an IoT. As one of the purposes of learning is to reduce the delay of an alarm message which is caused by an abnormality, such an algorithm must be decentralized, as the BS itself can have little to no information on the abnormality. Moreover, learning must be properly tailored to the resource constrained nature of MTDs.\\
\indent Consequently, we propose a distributed learning algorithm using which the MTDs can \emph{sequentially learn} whether there is an alarm message based on their own information about the true, underlying state as well as that gathered from other MTDs. Sequential learning is a process in which agents sequentially estimate an underlying, binary state $\theta = \{0, 1\}$ by observing their \emph{private signal} and the estimates of previous agents without a central, omniscient entity \cite{cover, fm}. For our model, $\theta = 1$ means that there is an alarm message, and $\theta = 0$ means that there is no alarm message. The private signal is a random variable, and its distribution depends on the true $\theta$. Moreover, the private signals are informative in estimating the true $\theta$ and are independently, identically distributed for different MTDs. Additionally, we assume that the likelihood ratio of any given private signal cannot be 0 or infinity. In other words, for our model, a private signal cannot fully reveal the true $\theta$ \cite{cover, fm}. Therefore, private signals can be interpreted as independent binary signals with some probability of inferring the true $\theta$. However, the agents do not know the probability of inferring the true $\theta$.
\subsection{Properties of Learning Parameters}
\indent The unique properties of the IoT pose challenges that are not dealt with in the conventional sequential learning literature such as \cite{cover, fm, infocas, herd}. To sequentially learn whether there is an alarm message, the private signal for an MTD is defined as its observation of an abnormality. In the presence of an abnormality, we let $p_{11}$ be the probability of observing the abnormality and, then, $p_{01} = 1-p_{11}$ would be the probability of not observing the abnormality. If there is no abnormality, the probability of observing an abnormality will be $p_{10}$, and the probability of not observing an abnormality will be $p_{00} = (1-p_{10})$. Without loss of generality, we assume $p_{11} > p_{10}$ and thus, $p_{01} < p_{00}$.\\
\indent Although the private signals are assumed to be equally informative in conventional sequential learning \cite{cover, fm}, the private signals for the MTDs are not necessarily equally informative or may not even be informative in some cases. This is because an abnormality can be observed with probability $p_{11}$ only within certain distance $r_d$ from the abnormality. Therefore, the MTDs located beyond a distance $r_d$ from the abnormality will observe the abnormality with probability $p_{10}$, and the private signals for the distant MTDs are not informative about $\theta$.\\
\indent In our model, the learning process of the MTDs is initiated by one of their \emph{neighbors}. The neighbors of an MTD are defined to be the MTDs that are within its communication range $r_c$. After an MTD has learned $\theta$, it will initiate the learning process on its neighbors.\footnote{Here, we assume that MTDs will use a control plane or signaling channel to share observations with their neighbors.}\\
\indent Conventional sequential learning algorithms are typically developed for one dimensional settings, such as a line of agents \cite{cover, fm}. However, the MTDs in an IoT are deployed in a two-dimensional network. This brings in new, unique challenges. For instance, one such challenge is the fact that it is possible that multiple learning sequences may simultaneously arrive at an MTD. We assume that the MTDs cannot determine which learning sequence is the best and thus, an MTD will randomly choose one of the learning sequences to follow. Moreover, it is possible for another learning sequence to arrive at an MTD that already belongs to some other learning sequence. In this case, we assume that an MTD will belong to the first learning sequence that it received.\\
\indent The necessary set of information for learning differs for the different types of sequential learning, and the set of information may include the estimates of $\theta$ of previous agents and the current state of learning process \cite{infocas, cover}. The two main types of sequential learning schemes are \emph{infinite memory} and \emph{finite memory}. The infinite memory sequential learning requires a complete knowledge of previous agents and their estimates of $\theta$, and the necessary set of information will grow infinitely as the learning progresses. The finite memory sequential learning only requires information about a finite set of previous agents and thus, the necessary set of information is finite and fixed. Most of works \cite{infocas, herd} are on infinite memory sequential learning, but the infinite memory sequential learning is not suitable for the resource limited MTDs due to computational limitations, power constraints, and communication restrictions. Therefore, we focus on developing a finite memory sequential learning for the MTDs in an IoT.
\subsection{Finite Memory Sequential Learning}
We propose a finite memory sequential learning algorithm, which extends the framework of \cite{cover} to the IoT case, thus enabling the MTDs to properly learn whether or not there is an alarm message. Our finite memory learning algorithm will nullify the effects of the limited observation range of an abnormality by using $K$ bits of information to learn $\theta$ and thus, the distant MTDs with uninformative private signals can learn to adjust their communication parameters. $K \geq 2$ is a design parameter affecting the convergence of learning, and the case of $K = 2$ will coincide with the finite memory learning algorithm introduced in \cite{cover}. We will focus on analyzing the performance of the proposed learning algorithm in the presence of an abnormality.\\
\indent In our algorithm, we use $K$ bits of information to capture $2$ bits for tracking the learning process and $(K-2)$ bits for the private signals of previous MTDs. The $K$ bits of information for an MTD $i$ are $\{e_{i-K+2}, \cdots, e_{i-1}, T_{i-1}, Q_{i-1}\}$, where $e_j$ is the private signal of MTD $j$, $T_{i-1}$ is current estimate of the true $\theta$, and $Q_{i-1}$ indicates whether $T_{i-1}$ should be changed. $T_{i-1}$ and $Q_{i-1}$ are the $2$ bits for tracking the learning process. After learning, $K$ bits of information that an MTD $i$ will transmit to its neighbors will be $\{e_{i-K+3}, \cdots, e_{i}, T_{i}, Q_{i}\}$, in which the oldest private signal $e_{i-K+2}$ is replaced by the private signal $e_{i}$ of an MTD $i$.\\
\indent Using previous private signals and its own private signal $\{e_{i-K+2}, \cdots, e_i\}$, an MTD will determine its \emph{private belief} by maximum likelihood estimation of $\theta$. The private belief is an estimate of the true $\theta$ only based on the private signals $\{e_{i-K+2}, \cdots, e_i\}$. In our proposed finite memory learning algorithm, the private belief will be used in place of the private signal in \cite{cover}. The private belief $x_i$ for MTD $i$ is:
\begin{align}
x_i &= \argmax\limits_\theta \Pr(e_{i-K+2}, \cdots, e_i \mid \theta),\\
&= \argmax\limits_\theta \prod\limits_{\scriptscriptstyle{j=i-K+2}}^{i} {\textstyle{\Pr(e_j \mid \theta)}} = \argmax\limits_\theta \prod\limits_{\scriptscriptstyle{j=i-K+2}}^{i} p_{\scriptstyle{1\theta}}^{\scriptstyle{e_j}} p_{\scriptstyle{0\theta}}^{\scriptstyle{1-e_j}},\\
&= \argmax\limits_\theta \sum\limits_{j=i-K+2}^{i} e_j \log(p_{\scriptstyle{1\theta}}) + (1-e_j) \log(p_{\scriptstyle{0\theta}}), \label{psml}
\end{align}
since the private signals are independent. If the MTDs know the values of $p_{11}$ and $p_{10}$, $x_i$ can easily be found using \eqref{psml}:
\begin{equation}
    x_{i}=\left\{
                \begin{array}{ll}
                  1 \hfill &\text{if} \ \frac{\sum\limits_{j=i-K+2}^{i} e_j \log(p_{\scriptstyle{11}}) + (1-e_j) \log(p_{\scriptstyle{01}})}{\sum\limits_{j=i-K+2}^{i} e_j \log(p_{\scriptstyle{10}}) + (1-e_j) \log(p_{\scriptstyle{00}})} \leq 1,\\
                  0 \hfill &\text{otherwise},
                \end{array}
              \right. \label{pslrt}
\end{equation}
which is based on likelihood ratio test. However, the MTDs may not know the values of $p_{11}$ and $p_{10}$ and cannot use \eqref{pslrt}. Here, it is assumed that the value of $p_{10}$ is negligible compared to the value of $p_{11}$. In such a case, using \eqref{pslrt}, the private belief $x_i$ of MTD $i$ will be:
\begin{equation}
    x_{i}=\left\{
                \begin{array}{ll}
                  1 \hfill &\text{if} \ \left(\sum\limits_{j=i-K+2}^{i} e_j\right) \geq 1,\\
                  0 \hfill &\text{otherwise}.
                \end{array}
              \right. \label{pslrt2}
\end{equation}
\indent Therefore, the private beliefs of the MTDs that do not know the values of $p_{11}$ and $p_{10}$ will be $1$ if any of the private signals $\{e_{i-K+2}, \cdots, e_i\}$ is 1. For an MTD $i$ within the observation range, the probability that $x_i$ is reflecting the true $\theta$ will be:
\begin{equation}
\Pr(x_i = 1 \mid \theta = 1) = 1 - p_{01}^{K-1} = 1 - (1 - p_{11})^{K-1}, \label{pb11}
\end{equation}
which is higher than the probability $p_{11}$ that a private signal is reflecting the true $\theta$. This implies that the private belief is more informative on the true $\theta$ than the private signal.\\
\indent For the MTDs outside of the observation range, some of their $(K-2)$ previous private signals may stem from the MTDs within the observation range. For an MTD $i$ outside of the observation range, we define $\kappa_i$ to be the number of previous private signals that are from MTDs within the observation range and $\eta_i = (K - 1 - \kappa_i)$ to be the number of private signals that are from the MTDs outside of the observation range. The probability that $x_i$ is reflecting the true $\theta$ for an MTD $i$ outside of the observation range will be:
\begin{equation}
\Pr(x_i = 1 \mid \theta = 1) = 1 - (1-p_{11})^{\kappa_i} (1-p_{10})^{\eta_i}.  \label{ptrueout}
\end{equation}
Under the same assumption about the values of $p_{11}$ and $p_{10}$, \eqref{ptrueout} will simplify to 
\begin{equation}
\Pr(x_i = 1 \mid \theta = 1) = 1 - (1-p_{11})^{\kappa_i}. \label{ptrueout2}
\end{equation}
\begin{proposition}
\normalfont The effective observation range $r_d'$ of an abnormality with $K$-bit finite memory learning is
\begin{equation}
r_d' = r_d + (K-2)r_c. \label{newrd}
\end{equation}
\end{proposition}
\begin{proof}
\indent From \eqref{ptrueout2}, $\kappa_i$ must be at least 1 for having 
\begin{equation}
\Pr(x_i = 1 \mid \theta = 1) > 0, 
\end{equation}
which implies that an MTD $i$ will have an informative private belief. Since the oldest private signal $e_{i-K+2}$ is replaced by the current private signal $e_i$, $\kappa_i$ of an MTD $i$ outside of the observation range $r_d$ will decrement by 1 each time the sequential learning progresses. Therefore, $\kappa_i$ will eventually be zero for some MTDs and thus, $\Pr(x_i = 1 \mid \theta = 1) = 0$ \eqref{ptrueout2} for such MTDs. This implies that the private belief $x_i$ for such MTDs will not be informative and thus, those MTDs will not learn the true $\theta$.\\
\indent Since there are $(K-2)$ bits of previous private signals, $\kappa_i$ will be zero after the sequential learning has progressed $(K-2)$ times. Therefore, the furthest MTD with informative private belief is located $(K-2)r_c$ from the observation range $r_d$.
\end{proof}
\indent The effective observation range $r_d'$ in \eqref{newrd} within which all of the MTDs can learn the true $\theta$ depends on $K$, which implies that $K$ is a design parameter that determines the number of MTDs that will learn the true $\theta$. Moreover, for $K > 2$, $r_d'$ is greater than $r_d$ and thus, $K$-bit finite memory learning can effectively mitigate the limited observation range, which greatly limits the number of MTDs that can learn the presence of an alarm message. Although more MTDs will learn with higher values of $K$, it is not ideal to choose $K$ to be arbitrarily high, because an arbitrarily high value of $K$ will have undesirable effects, such as system throughput reduction.
\section{Simulation Results and Analysis}
\begin{figure}[t]
\centering
\includegraphics[scale=.6]{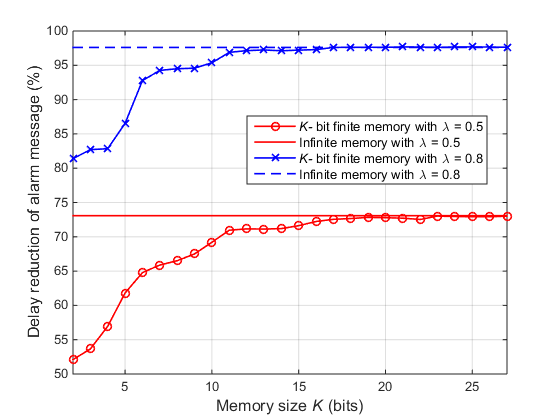}
\caption{Average percentage reduction of delay with respect to the case without any learning.}
\label{figdelay}
\end{figure}
For our simulations, we will focus on analyzing the performance of different learning algorithms under different MTD densities and the convergence of learning algorithms for different bits of memory. The learning algorithms that we will compare are infinite memory sequential learning and $K$-bit finite memory sequential learning for various integer values of $K \geq 2$. We assume that the MTDs are deployed on a square field with dimension $R$ in meters based on Poisson point process with some density parameter $\lambda$ in number of MTDs per square meter. Therefore, the expected number of MTDs $\mathbb{E}\left[N\right]$ is $R^2 \lambda$.\\
\indent In our simulations, we set the length of code $l$ to 4 bits and, the number of available codes $C$ is $15$. The time slot duration $\tau$ is 1 second, and the period $T\tau$ of periodic messages is $20$ seconds. The communication range $r_c$ between the MTDs is $2$ meters, while the observation range $r_d$ is $10$ meters. The dimension $R$ of the square field is $50$ meters. The value of $p_{11}$ is 0.80 and the value of $p_{10}$ is 0.001. We analyze the performance of learning algorithms with low MTD density of $\lambda = 0.5$ and high MTD density of $\lambda = 0.8$.\\
\indent Figure \ref{figdelay} shows the average percentage reduction of delay of alarm message with respect to expected delay without any learning algorithm as $K$ and $\lambda$ vary. The delay in \eqref{delay} of an alarm message is of critical interest, because an abnormality can be dealt with faster if the BS is notified quicker. Without any learning algorithm, the expected delays are $74.3$ seconds when $\lambda = 0.5$ and $988.2$ seconds when $\lambda = 0.8$. For both cases of $\lambda$, we can see significant percentage reductions of delay as $K$ increases, but the learning algorithms are more effective in reducing the delay when $\lambda$ is high. With $7$-bit finite memory learning, the delay is reduced by $66$\% when $\lambda = 0.5$ and $94$\% when $\lambda = 0.8$. Higher values of $K$ will reduce the delay more, but will also negatively affect the system throughput in \eqref{systhru} more as well.\\
\begin{figure}[t]
\centering
\includegraphics[scale=.6]{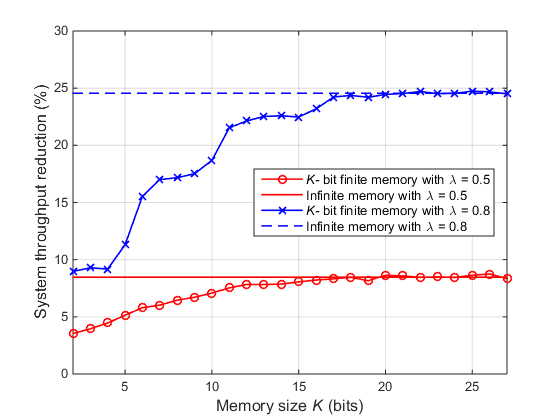}
\caption{Average percentage reduction of throughput with respect to the case with no learning.}
\label{figsysthru}
\end{figure}
\indent Figure \ref{figsysthru} shows the average percentage reduction of system throughput with respect to system throughput without any learning algorithm as $K$ and $\lambda$ vary. As $K$ increases, $K$-bit finite memory learning reduces the system throughput \eqref{systhru} more that the system throughput is reduced by $10$\% when $\lambda = 0.5$ and by $25$\% when $\lambda = 0.8$ for high values of $K$. Therefore, as shown in Figure \ref{delay} and Figure \ref{systhru}, there is a clear tradeoff between the percentage reductions of delay and system throughput that the higher percentage reduction in delay will lead to higher percentage reduction in system throughput. As $K$ is a design parameter, the value of $K$ can be chosen considering both percentage reductions in delay and system throughput to optimize the system performance.\\ 
\begin{figure}
\centering
\includegraphics[scale=.65]{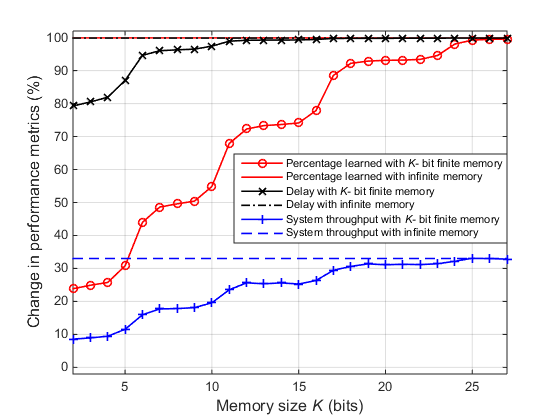}
\caption{Percentage learned, delay, and system throughput vs. memory size, for $\lambda = 0.8$}\vspace{-.4cm}
\label{complearned}
\end{figure}
\indent Figure \ref{complearned} shows the average percentage of MTDs that will learn the true $\theta$, the average percentage reduction of alarm message delay, and the average percentage reduction of system throughput as $K$ varies after all MTDs had chance to learn the true $\theta$. All percentages are measured with respect to the case with no learning. The number of MTDs that will learn the true $\theta$ increases as $K$ increases, because more MTDs will have informative private belief about true $\theta$ as the effective observation range $r_d'$ in \eqref{newrd}, which is an increasing function of $K$, increases. The percentage of MTDs that learned the true $\theta$ resembles a step function, because the learning algorithm has alternating phases to check whether the current estimate of $\theta$ should be $\theta = 0$ or $\theta = 1$.\\
\indent As more MTDs learn the presence of an alarm message, the alarm message delay will decrease, because an alarm message is more likely to successfully transmit as more MTDs learn to not transmit using the reserved code. Additionally, more MTDs with periodic messages will transmit with fewer available codes as more MTDs learn the presence of an alarm message. Since the periodic messages are less likely to successfully transmit with fewer available codes, the system throughput will also decrease. For $K \geq 11$, it is interesting to note that the system throughput reduces significant, while the alarm message delay does not improve noticeably. Similar to Figure \ref{figdelay} and Figure \ref{figsysthru}, there is a clear tradeoff between the percentage reductions in delay and system throughput. The optimal value of $K$ will depend on the required latency of an alarm message and throughput of the periodic messages.
\section{Conclusion}
In this paper, we have proposed a novel approach for enabling MTDs to communicate in the presence of both periodic and critical alarm messages. In particular, we have introduced a framework that enables MTDs with finite memory to learn whether or not rarely occurring alarm messages are in the system and, thus, adaptively adjust their transmission parameters. We have shown that such learning with finite memory is very well suited for the resource limited MTDs and also very effective against the unique limiting factors of learning of an IoT, such as limited observation range. Simulation results have shown that the delay of an alarm message, the system throughput, and the number of MTDs that learned the true $\theta$ are functions of $K$. Furthermore, the delay of an alarm message and the system throughput both decrease as $K$ increases and thus, $K$ must be chosen carefully to satisfy the given required performance metrics. One important future direction is to optimize the value of $K$ to meet specific system parameters of interest, such as system throughput and latency, in the presence of multiple types of messages.
\section*{Acknowledgment}
This research was supported by the U.S. Office of Naval Research (ONR) under Grant N00014-15-1-2709.


\begin{thebibliography}{99}
\bibitem{evans} D. Evans, ``The Internet of Things: How the Next Evolution of the Internet is Changing Everything," Cisco Internet Business Solutions Group, San Jose, CA, Apr. 2011.
\bibitem{saad} Z. Dawy, W. Saad, A. Ghosh, J. G. Andrews, and E. Yaacoub, ``Towards Massive Machine Type Cellular Communications," \textit{IEEE Wireless Communications Magazine}, to appear, 2016. Available: http://arxiv.org/ftp/arxiv/papers/1512/1512.03452.pdf.
\bibitem{ratasuk} R. Ratasuk, N. Mangalvedhe, A. Ghosh, and B. Vejlgaard, ``Narrowband LTE-M System for M2M Communication," in \textit{Proc. IEEE Vehicular Technology Conference}, Vancouver, Canada, Sept. 2014, pp. 1-5.
\bibitem{popov} J.J. Nielsen, G.C. Madueno, N.K. Pratas, R.B. Sørensen, C. Stefanovic\', and P. Popovski, ``What can wireless cellular technologies do about the upcoming smart metering traffic?" \textit{IEEE Communications Magazine}, vol. 53, no. 9, pp. 41-47, Sept. 2015.
\bibitem{latency} M. Weiner, M. Jorgovanovic, A. Sahai, and B. Nikolic\', ``Design of a Low-Latency, High-Reliability Wireless Communication System for Control Applications," in \textit{Proc. IEEE International Conference on Communications (ICC)}, Sydney, Australia, Jun. 2014, pp. 3829-3835.
\bibitem{dhil1} H.S. Dhillon, H.C. Huang, H. Viswanathan, and R.A. Valenzuela, ``Fundamentals of Throughput Maximization with Random Arrivals for M2M Communications," \textit{IEEE Transactions on Communications}, vol. 62, no. 11, pp. 4094-4109, Nov. 2014.
\bibitem{dhil2} H.S. Dhillon, H.C. Huang, H. Viswanathan, and R.A. Valenzuela, ``Power-Efficient System Design for Cellular-Based Machine-to-Machine Communications," \textit{IEEE Transactions on Wireless Communications}, vol. 12, no. 11, pp. 5740-5753, Nov. 2013.
\bibitem{cover} T.M. Cover, ``Hypothesis Testing with Finite Statistics," \textit{The Annals of Mathematical Statistics}, vol. 40, no. 3, pp. 828-835, Jun. 1969.
\bibitem{fm} K. Drakopoulos, A. Ozdaglar, and J. N. Tsitsiklis, ``On Learning with Finite Memory," \textit{IEEE Transactions on Information Theory}, vol. 59, no. 10, pp. 6859-6872, May 2013.
\bibitem{infocas} S. Bikhchandani, D. Hirshleifer, and I. Welch, ``A Theory of Fads, Fashion, Custom, and Cultural Changeas Informational Cascade," \textit{Journal of Political Economy}, vol. 100, no. 5, pp. 992-1026, Oct. 1992.
\bibitem{herd} A. V. Banerjee, ``A Simple Model of Herd Behavior," \textit{The Quarterly Journal of Economics}, vol. 107, no. 3, pp. 797-817, Aug. 1992.
\end{thebibliography}
\end{document}